\documentclass[journal]{IEEEtran}
\usepackage{amsmath,amssymb,amsfonts,mathrsfs,mathtools,bm, bbm, dsfont,mathrsfs,amsthm,blkarray}
\usepackage{cite}
\usepackage[dvipdfm]{graphicx}
\usepackage[dvipsnames]{xcolor}
\usepackage{algorithm}
\usepackage{algorithmic}

\usepackage{epsfig,epsf,psfrag,latexsym}

\usepackage{booktabs}
\usepackage{multirow,multicol}

\usepackage{xurl} 
\usepackage[breaklinks,colorlinks, linkcolor=MidnightBlue, anchorcolor=MidnightBlue, citecolor=MidnightBlue, urlcolor=MidnightBlue]{hyperref}
\usepackage[hyphenbreaks]{breakurl}
\usepackage{comment}

 \def\old#1{}    

\def\beq{\begin{equation}}
\def\eeq{\end{equation}}
\def\bea{\begin{eqnarray}}
\def\eea{\end{eqnarray}}
\def\ba{\begin{array}}
\def\ea{\end{array}}

\def\bitem{\begin{itemize}}
\def\eitem{\end{itemize}}
\def\ben{\begin{enumerate}}
\def\een{\end{enumerate}}

\def\eg{{\it e.g., \/}}

\newcommand{\beqa}{\begin{eqnarray}}
\newcommand{\eeqa}{\end{eqnarray}}
\newcommand{\beqan}{\begin{eqnarray*}}
\newcommand{\eeqan}{\end{eqnarray*}}

\newcounter{l1}
\newcounter{l2}
\newcounter{l3}
\newcommand{\bdotlist}{\begin{list}{$\bullet$}{}}
\newcommand{\bboxlist}{\begin{list}{$\Box$}{}}
\newcommand{\bbboxlist}{\begin{list}{\raisebox{.005in}{{\tiny
$\blacksquare$ \ \ }}}{}}
\newcommand{\bdashlist}{\begin{list}{$-$}{} }
\newcommand{\blist}{\begin{list}{}{} }
\newcommand{\barablist}{\begin{list}{\arabic{l1}}{\usecounter{l1}}}
\newcommand{\balphlist}{\begin{list}{(\alph{l2})}{\usecounter{l2}}}
\newcommand{\bAlphlist}{\begin{list}{\Alph{l2}.}{\usecounter{l2}}}
\newcommand{\bdiamlist}{\begin{list}{$\diamond$}{}}
\newcommand{\bromalist}{\begin{list}{(\roman{l3})}{\usecounter{l3}}}

\newtheorem{proposition}{Proposition}

\title{Risk-Based Capacity Accreditation of Resource- Colocated Large Loads in Capacity Markets}

\author{\large Siying Li, Lang Tong, and Timothy D. Mount
\thanks{\scriptsize
Siying Li and Lang Tong (\{sl2843, lt35\}@cornell.edu) are with the School of Electrical and Computer Engineering, Cornell University, Ithaca NY, USA. Timothy Mount (tdm2@cornell.edu) is with the Dyson School of Applied Economics and Management, Cornell University, Ithaca NY, USA.}
\thanks{\scriptsize This work was supported in part by the National Science Foundation under Award 2218110 and the Power Systems Engineering Research Center (PSERC) Project M-48.}
}

\begin{document}
\IEEEaftertitletext{\vspace{-1\baselineskip}}
\maketitle

\begin{abstract}
We study capacity accreditation of resource-colocated large loads, defined as large demands such as data center and manufacturing loads colocated with behind-the-meter generation and storage resources, synchronously connected to the bulk power system, and capable of participating in the wholesale electricity market as an integrated unit. Because the accredited capacity of a resource portfolio is not equal to the sum of its individual resources' capacity values, we adopt a risk-based capacity accreditation framework to evaluate the combined reliability contribution of colocated resources. Grounded in the effective load carrying capability (ELCC) metric, the proposed capacity accreditation employs a convex optimization engine that jointly dispatches colocated resources to minimize reliability risk. We apply the developed methodology to a hydrogen manufacturing facility with colocated renewable generation, storage, and fuel cell resources.
\end{abstract}


\section{Introduction}\label{sec:Intro}
The global electric power sector is experiencing an unprecedented surge in electricity demand. With an anticipated 3.5\% compound annual growth rate over the next 25 years, electricity consumption in 2050 is projected to exceed 2.5 times the 2023 level of 25,000 TWh. The top three sectors that fuel the demand growth are the ``large loads'' from data centers, hydrogen production, and electrified transportation \cite{McKinsey:24}. 

The large load surge has made resource adequacy one of the most pressing engineering, economic, and regulatory challenges. Resource adequacy means meeting demand with sufficient generation to ensure a prescribed level of reliability. Adding new generation capacities, however, requires long planning and permitting timelines, limiting its ability to respond to large and rapid demand growth. Although renewable generation can be added more quickly with low costs, its stochasticity results in low \emph{Effective Load Carrying Capability} (ELCC), and its integration requires substantial grid support. Already, sharp increases in renewable curtailments have been observed in recent years, signaling that renewable integration may be at the limit imposed by demand-generation mismatches, network constraints, and reliability requirements. 

A promising but underdeveloped approach to support resource adequacy amid rising large load is \emph{Resource Colocation}.\footnote{Both {\em colocation} and {\em collocation} are used in the literature to describe physically colocated resources, with {\em colocation} often emphasizing coordinated operation for operational or economic objectives. Some studies distinguish between colocated and hybrid resources based on metering and coordination; this paper does not make such a distinction.} In this context, a resource-colocated large load is a large load paired with behind-the-meter generation and storage resources that are synchronized with the bulk power system and participate in the wholesale market as a single entity.  Fig.~\ref{fig:hydrogen_system} illustrates a renewable-colocated hydrogen manufacturing facility with grid-connected renewable generation, an electrolyzer to produce hydrogen sold in a hydrogen market, and a hydrogen storage tank with the ability to generate power via fuel cell from stored hydrogen. Such a colocated hydrogen production resource can participate in the energy, ancillary service, and capacity markets. 

\begin{figure}[htbp]
    \centering
    \includegraphics[width=0.8\linewidth]{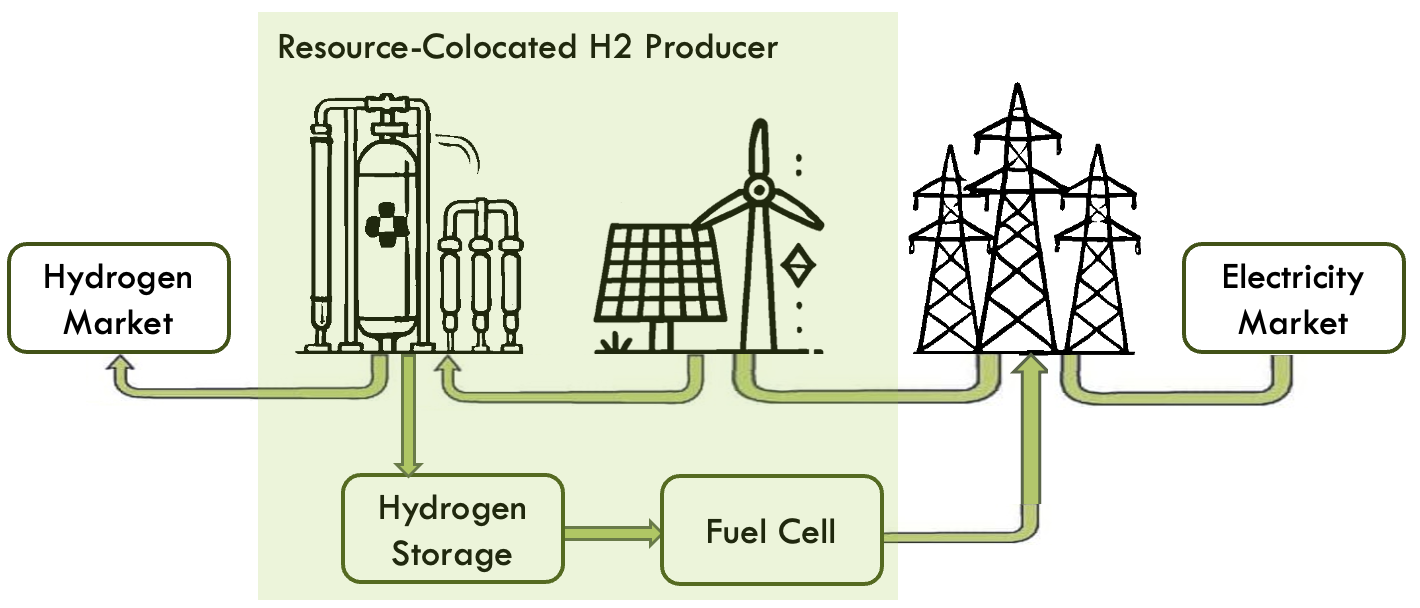}
    \vspace{-0.5em}
    \caption{Hydrogen production colocated with renewable, storage and fuel cell.}
    \label{fig:hydrogen_system}
\end{figure}

This paper focuses on resource-colocated large loads participating in the capacity market as a new type of capacity resource. In particular, a colocated large load can supply power to the grid as a generator, charge or discharge as a storage asset, and reduce consumption as a demand response resource. The optimal utilization of these capabilities to provide capacity support is the central objective of this work.

\subsection{Related Work}
Capacity accreditation is typically formulated within a risk-based framework. This framework maps a resource's nameplate capacity to its accredited value by accounting for system reliability risk. Accredited capacity is most commonly quantified using metrics such as ELCC or marginal reliability impact (MRI) \cite{Garver:1966,zhao2024}.

Most existing studies focus on the accreditation of individual resources \cite{Sioshansi14TPS, Edwards17SEGN, Evans19TPS}. With the ongoing evolution of power systems, increasing attention has been paid to portfolio effects in capacity accreditation. Several studies have examined the capacity value of colocated and hybrid resources, recognizing interactions among different resource types \cite{Schlag&etal:20E3Rpt,Stenclik&Goggin&Ela&Ahlstrom:ESIG22Rpt,Ericson&etal:22NREL,Wen&Song:2023}. In these works, the combined capacity value of multiple resources is assessed; however, the resources are generally not modeled as a single entity with internally coordinated dispatch decisions, nor are the internal operational constraints of colocated resources explicitly represented.
While the literature on colocated loads remains limited, related studies on colocated hydrogen production and data centers in energy and ancillary service markets have begun to emerge \cite{Glenk&Reichelstein:19Nature,Wu&etal:Energy22,Li&etal:25arxiv,Li&etal:26PESIM}. To the best of our knowledge, this paper is perhaps the first to examine the capacity accreditation of resource-colocated large loads in a capacity market.
 
Methodologies for capacity accreditation of colocated resources remain insufficiently developed. A commonly proposed portfolio-level approach is to sum the accredited capacities of individual components \cite{Stenclik&Goggin&Ela&Ahlstrom:ESIG22Rpt,Ericson&etal:22NREL}, which can yield inaccurate results due to interactions among resources. An alternative approach applies the risk-based accreditation framework at the aggregate level by dispatching the colocated resource according to extensions of existing scheduling rules \cite{Stenclik&Goggin&Ela&Ahlstrom:ESIG22Rpt,Ericson&etal:22NREL,Wen&Song:2023}. In the literature, such scheduling is typically driven by cost minimization or profit maximization in electricity markets \cite{Sioshansi14TPS}. To accurately capture reliability contributions, some studies adopt dispatch strategies that directly minimize reliability risk \cite{Edwards17SEGN, Evans19TPS}, which are most closely aligned with our approach. However, these methods often incur substantial computational burdens and therefore rely on simplifying assumptions—such as homogeneous storage efficiencies—to maintain tractability. Rule-based scheduling heuristics, such as those adopted in system operator practices \cite{PJM_Manual20A}, are computationally efficient, but as approximations of optimal coordination, they do not guarantee accurate capacity accreditation.

\subsection{Summary of Contribution}
The main technical contributions of this paper are threefold.
First, 
we propose a convex optimization-based computation engine that assesses system reliability through a reliability loss minimization formulation, in contrast to the rule-based heuristics commonly used in the literature and in practice. This formulation interprets existing heuristic practices as feasible, though generally suboptimal, solutions to the proposed optimization.
Second, in Sec.~\ref{sec:ELCC}, we introduce a scenario-based procedure for ELCC evaluation that accounts for stochastic resource availability and load variability.
Finally, in Sec. \ref{sec:Simulation}, we present a novel application of the proposed capacity accreditation methodology to a hydrogen manufacturing system colocated with on-site renewable generation, long-duration hydrogen storage, and fuel cell-based power production.

\section{Resource Modeling and ELCC}
\subsection{Resource Modeling}
We add {\em colocated resources (COL)} as a new resource type defined as a portfolio of commonly used resource types, which includes
(i) \emph{unlimited resources (U)}, such as conventional generators with dispatchable output subject to availability (\eg forced outage rates); (ii) \emph{variable resources (V)}, such as wind and solar, represented by stochastic generation profiles; (iii) \emph{storage resources (STO)}, modeled through state-of-charge dynamics with power limits and efficiency parameters; and (iv) \emph{flexible-demand resources (FLEX)}. 

Let $\mathcal{C}=\{\mathrm{U},\mathrm{V},\mathrm{STO},\mathrm{FLEX},\mathrm{COL}\}$ be the set of resource types.  Each resource $i$ of class $c\in\mathcal{C}$ is characterized by a pair $(Q_i^{c},\boldsymbol{\kappa}_i^{c})$ of parameters, where $Q_i^{c}$ denotes its nameplate capacity and $\boldsymbol{\kappa}_i^{c}$ denotes the corresponding parameter vector, collecting the associated technical and reliability parameters (\eg forced outage rates, storage efficiency and power limits, and demand reduction limits).
The complete set of capacity resources is parameterized by
$$
\boldsymbol{\theta}^{\mbox{\tiny\sf R}} := \big\{
\{ (Q_i^{c},\, \boldsymbol{\kappa}_i^{c}) \}_{i\in [n_c]} \big\}_{c\in\mathcal{C}},
$$
where $n_c$ denotes the number of units of class $c$. 

\subsection{ELCC for Colocated Resources}
By its standard definition \cite{Garver:1966}, ELCC quantifies the additional load\footnote{
Unlike flexible demand, which represents controllable demand-side resources, load refers to the system's inflexible electricity demand.} that a power system can supply while maintaining the same target reliability after adding a new resource. The calculation proceeds by first evaluating the reliability of a baseline system, then adding the resource and increasing the load until the reliability returns to the baseline level. The resulting load increase is the ELCC of the resource.

While the computation of ELCC for standalone resources is standard, extending it to colocated resources is nontrivial. One approach is to evaluate the ELCC of a colocated resource as the sum of the individual ELCC values. This simple approach, unfortunately, is incorrect as shown in the example below.

Consider a large load colocated with a 2 MW wind resource and a 5 MWh storage. ELCC is evaluated over four intervals, with the maximum unserved energy (UE) set to 4 MWh. The leftmost column of Fig.~\ref{fig:toy-example} shows two baseline electricity deficit profiles, each resulting in 4 MWh of UE.  
\vspace{-0.5em}
\begin{figure}[h]
    \centering
    \includegraphics[width=1.0\linewidth]{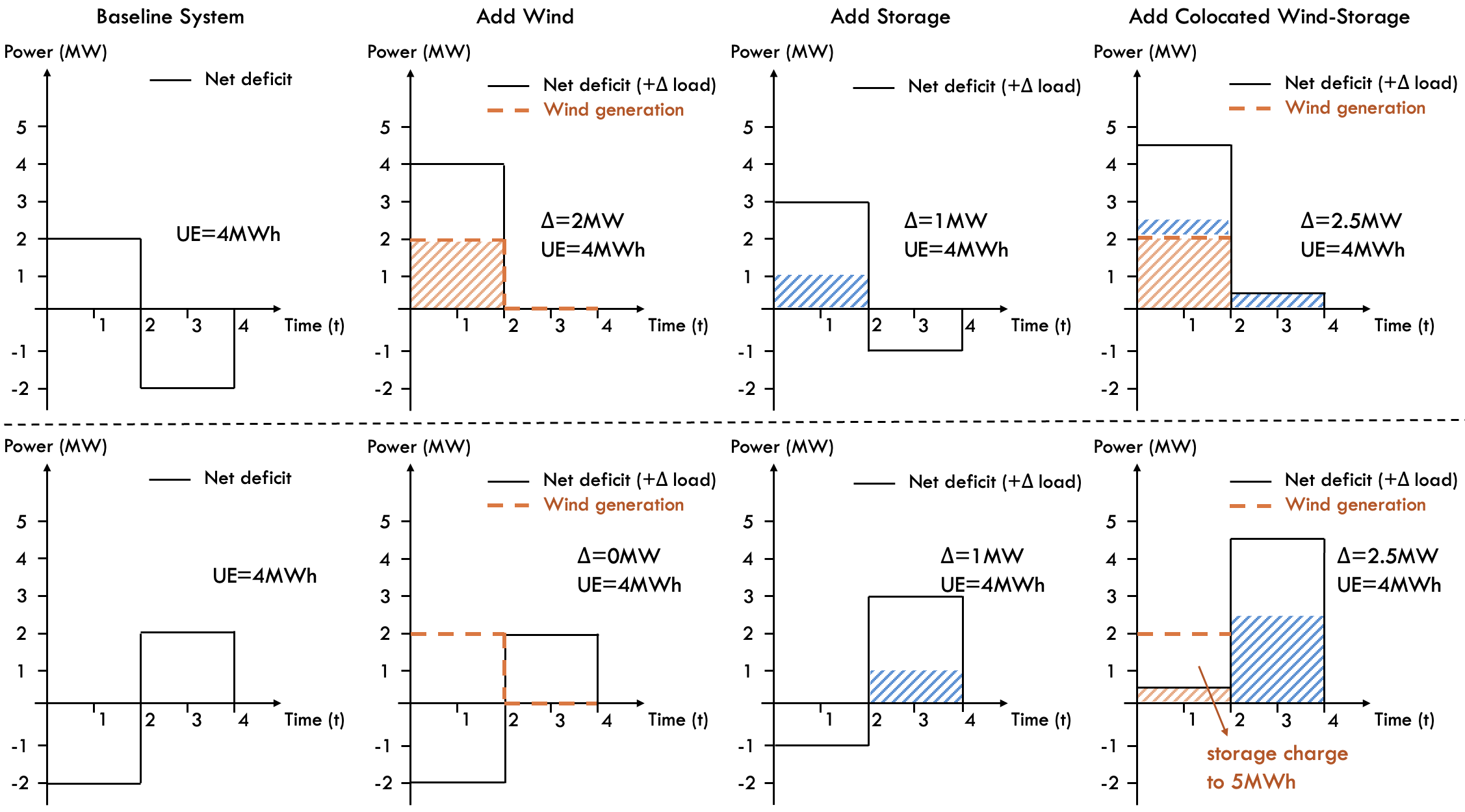}
    \vspace{-0.5em}
    \caption{Two toy examples of colocated wind-storage ELCC. Upper row: sub-additive case. Lower row: super-additive case.}
    \label{fig:toy-example}
\end{figure}
\vspace{-0.5em}

\bitem
\item Adding a 2 MW standalone wind resource (generation profiles shown in the second column) reduces UE differently in the two baselines: the top can accommodate an additional $\Delta=2$MW of load while keeping UE at 4 MWh, whereas the bottom cannot accommodate any additional load. Thus, the standalone wind ELCC is 2 MW in the top case and 0 MW in the bottom case. 
\item Adding a 5 MWh storage resource with 2 MWh initially stored (third column) allows 2 MWh of discharge in the first two intervals of the top baseline and in the last two intervals of the bottom baseline, permitting an additional $\Delta=1$ MW of load while maintaining the UE level. The ELCC for the storage is therefore 1 MW.
\item When the 2 MW wind and 5 MWh storage are colocated, optimal storage scheduling enables an additional $\Delta=2.5$ MW of load in both baselines. The coloctaed ELCC is thus $\Delta=2.5$ MW. 
\eitem 
We see that the sum of wind and storage (standalone) ELCCs is higher than the colocated ELCC (under optimal scheduling) in the top baseline (referred to as sub-additive) and lower in the bottom case (referred to as super-additive).

\section{Reliability Evaluation with Optimal Dispatch}\label{sec:PSReliability}
The example above shows that the ELCC evaluation of a colocated resource requires optimizing over the resource portfolio to jointly reduce reliability risk. This section develops that optimization model, which serves as the computation engine for ELCC computation described in the next section. 

\subsection{Reliability Loss Minimization}\label{sec:opt_formulation}
Consider a power system characterized by the resource vector $\boldsymbol{\theta}^{\mbox{\tiny\sf R}}$ and the load vector $\boldsymbol{\theta}^{\mbox{\tiny\sf D}}$. 
Here, $\boldsymbol{\theta}^{\mbox{\tiny\sf D}} := (Q^{\mbox{\tiny\sf D}}, \{\boldsymbol{\kappa}_i^{\mbox{\tiny\sf D}}\}_{i\in[n_{\text{\fontsize{4pt}{4pt}\selectfont \sf D}}]})$ 
collects the peak load level $Q^{\mbox{\tiny\sf D}}$ and the parameters that describe the temporal and spatial profiles of each load component (\eg typical daily shapes).

Each scenario $\omega$ represents a joint realization of resource availability and load. Given $\omega$, the model specifies $\boldsymbol{P}^{\mbox{\tiny\sf U}}(\boldsymbol{\theta}^{\mbox{\tiny\sf R}},\omega)$, $\boldsymbol{P}^{\mbox{\tiny\sf V}}(\boldsymbol{\theta}^{\mbox{\tiny\sf R}},\omega)$, $\boldsymbol{s}(\boldsymbol{\theta}^{\mbox{\tiny\sf R}},\omega)$, $\bar{\boldsymbol{P}}^{\mbox{\tiny\sf FLEX}}(\boldsymbol{\theta}^{\mbox{\tiny\sf R}},\omega)$, $\bar{\boldsymbol{P}}^{\mbox{\tiny\sf COL}}(\boldsymbol{\theta}^{\mbox{\tiny\sf R}},\omega)$, and $\boldsymbol{D}(\boldsymbol{\theta}^{\mbox{\tiny\sf D}},\omega)$,
where $\boldsymbol{P}^{\mbox{\tiny\sf U}}=\{P_t^{\mbox{\tiny\sf U}}\}_{t\in[T]}$ and $\boldsymbol{P}^{\mbox{\tiny\sf V}}=\{P_t^{\mbox{\tiny\sf V}}\}_{t\in[T]}$ are time-indexed total generation trajectories for conventional and renewable units,
$\boldsymbol{s}=\{s_i\}_{i\in[n_{\text{\fontsize{4pt}{4pt}\selectfont \sf STO}}]}$ gives the initial state of charge for each storage unit,
$\bar{\boldsymbol{P}}^{\mbox{\tiny\sf FLEX}}=\{\bar{P}_{i,t}^{\mbox{\tiny\sf FLEX}}\}_{i\in[n_{\text{\fontsize{4pt}{4pt}\selectfont \sf FLEX}}],\,t\in[T]}$ and $\bar{\boldsymbol{P}}^{\mbox{\tiny\sf COL}}=\{\bar{P}_{i,t}^{\mbox{\tiny\sf COL}}\}_{i\in[n_{\text{\fontsize{4pt}{4pt}\selectfont \sf COL}}],\,t\in[T]}$ denote the baseline power trajectories representing the nominal operating behavior of each flexible-demand resource $i$ and each colocated asset portfolio $j$ at time $t$, and
$\boldsymbol{D}=\{D_t\}_{t\in[T]}$ denotes the time-varying system load.
For brevity, the dependence on $(\boldsymbol{\theta},\omega)$ is omitted here and in subsequent expressions.

At each time period $t\in[T]$, the optimization determines the operation of flexible resources and system-level balancing variables, collectively denoted by
\vspace{-0.1em}
$$
\boldsymbol{P}_t \coloneqq \left\{
\begin{aligned}
& \{P^{\mbox{\tiny\sf CH}}_{i,t},\, P^{\mbox{\tiny\sf DIS}}_{i,t},\, e_{i,t}\}_{i\in[n_{\text{\fontsize{4pt}{4pt}\selectfont \sf STO}}]},\\
& \{P^{\mbox{\tiny\sf FLEX}}_{i,t},\, P^{\mbox{\tiny\sf RED}}_{i,t}\}_{i\in[n_{\text{\fontsize{4pt}{4pt}\selectfont \sf FLEX}}]},\\
& \{\boldsymbol{P}^{\mbox{\tiny\sf COL}}_{i,t}\}_{i\in[n_{\text{\fontsize{4pt}{4pt}\selectfont \sf COL}}]}, P_t^{\mbox{\tiny\sf CUR}},\, P_t^{\mbox{\tiny\sf UE}}
\end{aligned}
\right\},
$$
where $P^{\mbox{\tiny\sf CH}}_{i,t}$, $P^{\mbox{\tiny\sf DIS}}_{i,t}$, and $e_{i,t}$ denote the charging power, discharging power, and state of charge of storage $i$ at time $t$. $P^{\mbox{\tiny\sf FLEX}}_{j,t}$ and $P^{\mbox{\tiny\sf RED}}_{i,t}$ denote the actual power consumption and demand reduction of flexible demand $i$ at time $t$. $\boldsymbol{P}^{\mbox{\tiny\sf COL}}_{i,t}$ is a generic vector variable representing the operational decisions of colocated resource $i$ at time $t$.
$P_t^{\mbox{\tiny\sf CUR}}$ and $P_t^{\mbox{\tiny\sf UE}}$ are system curtailed and unserved power, respectively.

Let $\boldsymbol{P} := \{\boldsymbol{P}_t\}_{t\in[T]}$ include all decision variables over the horizon.
For a given scenario~$\omega$, the reliability loss associated with the dispatch trajectory is denoted by
$R(\boldsymbol{P};\,\boldsymbol{\theta}^{\mbox{\tiny\sf R}},\boldsymbol{\theta}^{\mbox{\tiny\sf D}},\omega)$, and the system reliability assessment is formulated as the following optimization problem.
\begin{subequations}\label{eq:opt_reliability}
\begin{align}
& \underset{\substack{\boldsymbol{P}}}{\rm minimize} && R(\boldsymbol{P};\boldsymbol{\theta}^{\mbox{\tiny\sf R}},\boldsymbol{\theta}^{\mbox{\tiny\sf D}},\omega) \\ \nonumber
& \mbox{subject to} && \forall t \in [T],~ \forall i \in [n_{\text{\fontsize{4pt}{4pt}\selectfont \sf STO}}], ~\forall j \in [n_{\text{\fontsize{4pt}{4pt}\selectfont \sf FLEX}}],~\forall r \in [n_{\text{\fontsize{4pt}{4pt}\selectfont \sf COL}}],\\  \label{eq:balance1}
&&& G_t=P_t^{\mbox{\tiny\sf U}}+P_t^{\mbox{\tiny\sf V}}+\mathbf{1}^\top\boldsymbol{P}^{\mbox{\tiny\sf DIS}}_t+P_t^{\mbox{\tiny\sf UE}}+g(\boldsymbol{P}_{t}^{\mbox{\tiny\sf COL}}),\\
&&& L_t=D_t+\mathbf{1}^\top \boldsymbol{P}^{\mbox{\tiny\sf CH}}_t+\mathbf{1}^\top\boldsymbol{P}^{\mbox{\tiny\sf FLEX}}_t+P^{\mbox{\tiny\sf CUR}}_t+l(\boldsymbol{P}_{t}^{\mbox{\tiny\sf COL}}),\\   \label{eq:balance2}
&&& G_t=L_t,\\ \label{eq:bilinearcons}
&&& P_{i,t}^{\mbox{\tiny\sf CH}}P_{i,t}^{\mbox{\tiny\sf DIS}}=0,\\ \label{eq:storage_limits}
&&& 0\leq P_{i,t}^{\mbox{\tiny\sf CH}}\leq \bar{P}_{i}^{\mbox{\tiny\sf CH}},~~ 0\leq P_{i,t}^{\mbox{\tiny\sf DIS}}\leq \bar{P}_{i}^{\mbox{\tiny\sf DIS}},\\ \label{eq:storage_evolve}
&&& e_{i,1}=s_i,~~ \underline{E}_i \leq e_{i,t+1}\leq \overline{E}_i,\\  \label{eq:storage_dynamics}
&&& e_{i,t+1}=e_{i,t}+\eta_i P_{i,t}^{\mbox{\tiny\sf CH}}-P_{i,t}^{\mbox{\tiny\sf DIS}},\\ \label{eq:flexible_demand1}
&&& P_{j,t}^{\mbox{\tiny\sf FLEX}} = \bar{P}_{j,t}^{\mbox{\tiny\sf FLEX}}-P_{j,t}^{\mbox{\tiny\sf RED}},\\ \label{eq:flexible_demand2}
&&& 0\leq P_{j,t}^{\mbox{\tiny\sf RED}}\leq \bar{P}_{j,t}^{\mbox{\tiny\sf RED}} \mathbbm{1}_{\{P^{\mbox{\tiny\sf U}}_t+P^{\mbox{\tiny\sf V}}_t< D_t+\mathbf{1}^\top\bar{\boldsymbol{P}}^{\mbox{\tiny\sf FLEX}}_t+\mathbf{1}^\top\bar{\boldsymbol{P}}_{t}^{\mbox{\tiny\sf COL}}  \}},\\  \label{eq:flexible_demand3}
&&& \mathbf{1}^\top\!\!\boldsymbol{P}_{t}^{\mbox{\tiny\sf RED}}\!\!\leq\! \max(\!D_t\!+\!\!\mathbf{1}^\top\!\!\bar{\boldsymbol{P}}^{\mbox{\tiny\sf FLEX}}_t\!\!+\!\mathbf{1}^\top\!\!\bar{\boldsymbol{P}}_{t}^{\mbox{\tiny\sf COL}}\!\!\!-\!\!(\!P^{\mbox{\tiny\sf U}}_t\!+\!P^{\mbox{\tiny\sf V}}_t\!), 0),\\ \label{eq:colocated_resource}
&&& \boldsymbol{F}_r^{\mbox{\tiny\sf COL}}(\boldsymbol{P}_{r,t}^{\mbox{\tiny\sf COL}})\leq \mathbf{0}.
\end{align}
\end{subequations}

Depending on the chosen reliability metric, $R(\cdot)$ represents different measures of reliability loss, such as:\footnote{Here, $R(\cdot)$ gives the per-scenario reliability loss. The corresponding reliability metric is obtained by taking the expectation across all stochastic scenarios; See Sec.~\ref{sec:ELCC}.}
\begin{itemize}
    \item For loss of load expectation (LOLE) minimization: 
    \begin{equation}
    R(\boldsymbol{P};\boldsymbol{\theta}^{\mbox{\tiny\sf R}},\boldsymbol{\theta}^{\mbox{\tiny\sf D}},\omega) = \sum_{t=1}^T \mathbbm{1}_{\{P_t^{\mbox{\tiny\sf UE}}>0\}},
    \end{equation}
    \item For expected unserved energy (EUE) minimization: 
    \begin{equation}\label{eq:eue_objective}
    R(\boldsymbol{P};\boldsymbol{\theta}^{\mbox{\tiny\sf R}},\boldsymbol{\theta}^{\mbox{\tiny\sf D}},\omega) = \sum_{t=1}^T P_t^{\mbox{\tiny\sf UE}}.
    \end{equation}
\end{itemize}

We define $G_t$ as the total power supply at time $t$, including generation, storage discharge, and colocated output $g(\boldsymbol{P}_{t}^{\mbox{\tiny\sf COL}})$. $L_t$ denotes the total demand, comprising inflexible load and the actual consumption of flexible-demand and colocated resources $l(\boldsymbol{P}_{t}^{\mbox{\tiny\sf COL}})$. 
For notational convenience, $g(\cdot)$ and $l(\cdot)$ denote the aggregate colocated output and consumption.
Equations \eqref{eq:balance1}--\eqref{eq:balance2} enforce power balance, with shortfalls represented by unserved energy $P_t^{\mbox{\tiny\sf UE}}$ and surpluses by curtailment $P_t^{\mbox{\tiny\sf CUR}}$. Equations \eqref{eq:bilinearcons}--\eqref{eq:storage_dynamics} specify storage operations, including charge/discharge limits and state-of-charge dynamics.
Flexible-demand constraints \eqref{eq:flexible_demand1}--\eqref{eq:flexible_demand3} ensures that demand reductions occur only under supply shortages and apply solely to covering the shortfall, consistent with industry practices \cite{PJM_Manual20A}.
\eqref{eq:colocated_resource} defines a generic constraint set representing the operational characteristics of colocated resources.

The scheduling in this formulation differs from economic dispatch, as the objective is to minimize reliability loss rather than generation cost. However, this distinction matters only under scarcity conditions: when resources are abundant and no reliability loss occurs, applying economic dispatch yields the same reliability outcomes.
The rule-based scheduling commonly used in capacity accreditation practices corresponds to feasible solutions of this optimization.

\subsection{Convex Reformuation of Reliability Loss Minimization}
The optimization in \eqref{eq:opt_reliability} is generally a mixed-integer nonlinear program (MINLP), mainly due to the bilinear storage operation constraints in \eqref{eq:bilinearcons} and the indicator functions in \eqref{eq:flexible_demand2}.\footnote{Colocated resources, modeled as portfolios of standard resource types, are assumed not to introduce additional nonconvexities beyond those already represented.} These nonconvexities make large-scale reliability evaluation computationally challenging, especially under numerous stochastic scenarios.

Nevertheless, tractable reformulations can be developed for certain reliability metrics.
In particular, when the EUE metric is adopted, the objective \eqref{eq:eue_objective} is linear in the decision variables, enabling a convex reformulation of \eqref{eq:opt_reliability} through appropriate relaxations and transformations, as shown below.\footnote{If the objective function is nonconvex, suitable relaxations can yield tractable approximations.}
\begin{proposition}
    If the objective function $R(\cdot)$ is convex, \eg as in \eqref{eq:eue_objective} for EUE minimization, then the reliability optimization problem \eqref{eq:opt_reliability} can be equivalently reformulated as the following convex program:
    \vspace{-0.5em}
    \begin{equation}\label{eq:cvx}
    \begin{aligned}
    & \underset{\substack{\{P_{i,t}^{\mbox{\tiny\sf CH}}, P_{i,t}^{\mbox{\tiny\sf DIS}},e_{i,t},\\\boldsymbol{P}^{\mbox{\tiny\sf COL}}_{r,t},P_t^{\mbox{\tiny\sf CUR}},P_t^{\mbox{\tiny\sf UE}} \}}}{\rm minimize} && \sum_{t=1}^{T} P_t^{\mbox{\tiny\sf UE}} \\
    & \text{subject to} && \forall t \in [T],~~ \forall i \in [n_{\text{\fontsize{4pt}{4pt}\selectfont \sf STO}}],~\forall r \in [n_{\text{\fontsize{4pt}{4pt}\selectfont \sf COL}}], \\
    & && \eqref{eq:balance1}, \\
    & && G_t\!
      =\!D_t\!+\!\mathbf{1}^\top \!\boldsymbol{P}^{\mbox{\tiny\sf CH}}_t\!+\!\mathbf{1}^\top\!\tilde{\boldsymbol{P}}^{\mbox{\tiny\sf FLEX}}_t\!+\!P^{\mbox{\tiny\sf CUR}}_t\!+\!l(\boldsymbol{P}_{t}^{\mbox{\tiny\sf COL}}), \\  
    & && \eqref{eq:storage_limits},~\eqref{eq:storage_evolve},~\eqref{eq:storage_dynamics}, \eqref{eq:colocated_resource}.
    \end{aligned}
    \end{equation}
    The aggregate actual power consumption of flexible-demand resources, denoted by $\mathbf{1}^\top\tilde{\boldsymbol{P}}^{\mbox{\tiny\sf FLEX}}_t$, is given by
    \begin{equation}
    \mathbf{1}^\top\tilde{\boldsymbol{P}}^{\mbox{\tiny\sf FLEX}}_t=\left\{
    \begin{array}{lcl}
    \mathbf{1}^\top\bar{\boldsymbol{P}}^{\mbox{\tiny\sf FLEX}}_t && S_t\leq 0,\\
    \min(\mathbf{1}^\top\bar{\boldsymbol{P}}^{\mbox{\tiny\sf RED}}_t, S_t) && S_t>0.\\    
    \end{array}\right.
    \end{equation}
    where $S_t\!=\!D_t\!+\!\mathbf{1}^\top\bar{\boldsymbol{P}}^{\mbox{\tiny\sf FLEX}}_t\!+\!\mathbf{1}^\top\bar{\boldsymbol{P}}_{t}^{\mbox{\tiny\sf COL}}\!-\!P^{\mbox{\tiny\sf U}}_t\!-\!P^{\mbox{\tiny\sf V}}_t$ denotes the net shortfall between total demand and the available uncontrollable generation at time $t$ before scheduling flexible resources.
\end{proposition}
\begin{proof}[Sketch of Proof]
The convexification can be established in two steps. First, the bilinear storage constraints in \eqref{eq:bilinearcons} are relaxed without loss of optimality. By allowing simultaneous charging and discharging, any solution with $\tilde{P}_{i,t}^{\mbox{\tiny\sf CH}}>0$ and $\tilde{P}_{i,t}^{\mbox{\tiny\sf DIS}}>0$ can be replaced by an alternative solution involving only net charging or discharging, which yields the same state-of-charge trajectory and an objective value that is no worse. The objective strictly improves if the storage round trip efficiency is less than 100\%.

Second, the flexible demand constraints in \eqref{eq:flexible_demand1}--\eqref{eq:flexible_demand3} can be reformulated to directly determine the aggregate actual consumption of flexible-demand resources based on system conditions. Specifically, when there is no generation shortfall ($S_t\leq 0$), flexible demand operates at its maximum consumption level; when a shortfall occurs ($S_t>0$), the total demand reduction is set to the minimum of the available reduction capacity and the shortfall amount.
\end{proof}

Note that in the optimization, generation and load are represented through aggregated profiles. Since this formulation does not consider network constraints, other resources can similarly be scheduled by class to enhance computational efficiency. Specifically, storage units and flexible-demand resources can each be grouped into classes according to their technical characteristics, consistent with class-based scheduling practices used by system operators.

\subsection{Extension to Network Constraints}
Incorporating network constraints introduces locational considerations into the reliability evaluation. Transmission limitations influence the deliverability of generation resources and, consequently, the overall system reliability. The optimization formulation in \eqref{eq:opt_reliability} can be extended to include network effects by integrating DC power flow equations.

While this extension provides a more accurate representation of system reliability, it also increases the computational complexity. In particular, generation and load profiles can not be aggregated, and other resources can no longer be scheduled by class without losing the locational information required for modeling power flows. To balance modeling fidelity and computational efficiency, one practical approach is to perform parallel locational reliability evaluations for different system areas, under the assumption of sufficient inter-area transmission capacity.

\section{Risk-based Capacity Accreditation}\label{sec:ELCC}
Building on the reliability evaluation framework developed in Sec.~\ref{sec:PSReliability}, we now apply it to quantify the ELCC.
Consider a baseline system characterized by $(\boldsymbol{\theta}^{\mbox{\tiny\sf R}}, \boldsymbol{\theta}^{\mbox{\tiny\sf D}})$.
The expected system reliability, accounting for the stochastic scenarios of both resources and load, is evaluated as
\begin{equation}
\mathcal{R}(\boldsymbol{\theta}^{\mbox{\tiny\sf R}}, \boldsymbol{\theta}^{\mbox{\tiny\sf D}}) 
:= \mathbb{E}_{\omega}\!\left[ R(\boldsymbol{P}^*;\boldsymbol{\theta}^{\mbox{\tiny\sf R}},\boldsymbol{\theta}^{\mbox{\tiny\sf D}},\omega) \right],
\end{equation}
where $\boldsymbol{P}^*$ denotes the optimal dispatch decisions obtained from \eqref{eq:opt_reliability}.\footnote{Alternative measures such as conditional value at risk (CVaR) can also be used to capture reliability risks.}

Let $(Q', \boldsymbol{\kappa}')$ denote the capacity and associated technical and reliability parameters of a new resource to be added to the system. 
The augmented system resource set is expressed as
$$
\boldsymbol{\theta}^{\mbox{\tiny\sf R}}_{+}(Q', \boldsymbol{\kappa}')
:= \boldsymbol{\theta}^{\mbox{\tiny\sf R}} \oplus (Q', \boldsymbol{\kappa}'),
$$
where $\oplus$ represents the augmentation of the existing resource set with the new unit.

Following Garver's classical definition \cite{Garver:1966}, system load is represented by the annual peak load $Q^{\mbox{\tiny\sf D}}$.
An increase in load by $\Delta$ corresponds to raising the peak load by $\Delta$ while scaling the load profile proportionally:
$$
\boldsymbol{\theta}^{\mbox{\tiny\sf D}}(\Delta)
:= \big(Q^{\mbox{\tiny\sf D}}+\Delta,\,
\{\boldsymbol{\kappa}_i^{\mbox{\tiny\sf D}}\}_{i\in[n_{\text{\tiny D}}]}\big).
$$

The ELCC of the new resource $(Q', \boldsymbol{\kappa}')$ is the maximum additional load $\Delta$ that can be served without degrading the target reliability level, formally defined as
$$
{\rm ELCC}(Q', \boldsymbol{\kappa}')
:= \Delta ~
\text{s.t.}~
\mathcal{R}(\boldsymbol{\theta}^{\mbox{\tiny\sf R}}, \boldsymbol{\theta}^{\mbox{\tiny\sf D}})
= \mathcal{R}\big(\boldsymbol{\theta}^{\mbox{\tiny\sf R}}_+(Q', \boldsymbol{\kappa}'), \boldsymbol{\theta}^{\mbox{\tiny\sf D}}(\Delta)\big).
$$
The specific computational procedure for determining $\Delta$ is presented in Algorithm~\ref{alg1}.
\begin{algorithm}[H]
\caption{Scenario-based ELCC Computation}\label{alg1}
\begin{algorithmic}[1]
\REQUIRE Baseline system $(\boldsymbol{\theta}^{\mbox{\tiny\sf R}}, \boldsymbol{\theta}^{\mbox{\tiny\sf D}})$, 
new resource $(Q',\boldsymbol{\kappa}')$, 
scenario generator $\mathsf{Gen}(\cdot)$, 
number of scenarios $N$, tolerance $\varepsilon$.
\ENSURE $\Delta = {\rm ELCC}(Q',\boldsymbol{\kappa}')$.

\STATE \textbf{Scenario generation:} 
Generate and fix $N$ stochastic scenarios 
$\{\omega^{(n)}\}_{n=1}^{\mbox{\tiny N}} \leftarrow \mathsf{Gen}(\boldsymbol{\theta}^{\mbox{\tiny\sf R}}, \boldsymbol{\theta}^{\mbox{\tiny\sf D}})$, representing realizations of resource availability and load.%
\footnotemark[1]

\STATE \textbf{Baseline evaluation (Step \ref{step:base}):}
\FOR{$n = 1,\dots,N$}\label{step:base}
    \STATE Solve~\eqref{eq:opt_reliability} under scenario~$\omega^{(n)}$ 
    to obtain $\boldsymbol{P}^{\ast}$.
    \STATE Compute $R_{\mathrm{base}}^{(n)} \gets 
    R(\boldsymbol{P}^{\ast};\boldsymbol{\theta}^{\mbox{\tiny\sf R}},\boldsymbol{\theta}^{\mbox{\tiny\sf D}},\omega^{(n)})$.
\ENDFOR
\STATE $\widehat{\mathcal{R}}_0 \gets \frac{1}{N}\sum_{n=1}^{\mbox{\tiny N}} R_{\mathrm{base}}^{(n)}$.
\vspace{3pt}

\STATE \textbf{Reliability matching:}
\STATE Set $\boldsymbol{\theta}^{\mbox{\tiny\sf R}}_{+}\gets\boldsymbol{\theta}^{\mbox{\tiny\sf R}}\oplus(Q',\boldsymbol{\kappa}')$.
\REPEAT
  \STATE \emph{Repeat Step \ref{step:base}} with $\big(\boldsymbol{\theta}^{\mbox{\tiny\sf R}}_{+}, \boldsymbol{\theta}^{\mbox{\tiny\sf D}}(\Delta)\big)$ and the fixed $\{\omega^{(n)}\}$ to obtain $\widehat{\mathcal R}(\Delta)$.
  \STATE Update $\Delta$ to reduce $|\widehat{\mathcal R}(\Delta)-\widehat{\mathcal R}_0|$.\footnotemark[2]
\UNTIL{$|\widehat{\mathcal R}(\Delta)-\widehat{\mathcal R}_0|<\varepsilon$}

\STATE \textbf{Output:} Return $\Delta$ as 
${\rm ELCC}(Q', \boldsymbol{\kappa}')$.
\end{algorithmic}
\end{algorithm}
\footnotetext[1]{%
$\mathsf{Gen}(\cdot)$ can be any Monte Carlo or bootstrapping routine that samples resource and load trajectories consistent with system parameters. }
\footnotetext[2]{%
The iterative search for $\Delta$ can employ monotone root-finding methods such as bisection or secant search, leveraging the monotonicity of reliability loss with respect to load.}

In capacity accreditation applications, system operators often compute the ELCC of a resource relative to different benchmarks or reference units. Examples include defining ELCC as the equivalent increase in load that can be supplied at the target reliability level, as the capacity of a perfectly reliable generator, or as that of a reference unit with a specified forced outage rate \cite{NERC_2011}. The proposed approach is general and can be applied to any of these benchmark-based calculations.

\section{ELCC of Resource-Colocated Hydrogen}\label{sec:Simulation}
\noindent{\textbf{System description:}} We studied a hydrogen manufacturing system colocated with renewable generation, electrolyzers, and hydrogen storage with fuel cell, as illustrated in Fig.~\ref{fig:hydrogen_system}. 
The electrolyzer can modulate its power consumption within a specified operating range and convert electricity into hydrogen, which can either be sold directly or stored in a hydrogen tank. The ability to adjust power drawn from the grid makes the hydrogen producer a flexible-demand resource. 

The system can also export power to the grid using colocated renewable generation or by running the fuel cell with hydrogen stored in the tank. However, because the round-trip efficiency of converting electricity to hydrogen and back is typically quite low, it is generally suboptimal to use stored hydrogen for power production during regular operating hours. Instead, the stored hydrogen is dispatched during scarcity events, when electricity prices are sufficiently high. 

\noindent{\textbf{Case study setup:}}
In the case study, the colocated wind capacity was 14~MW, the nominal electrolyzer demand was 14~MW with up to 30\% demand response capability, and the stored hydrogen could generate up to 11.2~MW of electricity. 
The analysis was performed on a calibrated baseline power system with a 650~MW peak load, 640~MW of conventional generation capacity, 300~MW of wind, 100~MW of solar, and a total of 700~MWh of energy storage. The load and renewable generation profiles were derived from publicly available NYISO real-time data \cite{NYISO:Data}.

\noindent{\textbf{ELCC evaluation under different reliability metrics:}}
We evaluated resource-specific and system-level ELCC for the colocated hydrogen system under LOLE- and EUE-based reliability criteria. Here, the qualified capacity was defined as the nameplate capacity for renewable, the reducible MW of flexible demand, the power rating for storage, and the sum of these for the colocated system. All reported percentages are relative to these definitions.

As shown in Fig.~\ref{fig:EUE_vs_LOLE} (left), 
under LOLE minimization, wind achieved an ELCC of 7.8~kW, representing less than 0.1\% of its qualified capacity when participating as a single resource in the capacity market. Flexible demand and storage reached ELCCs of 4.2~MW (100\%) and 9.7~MW (86.7\%), respectively.
The colocated system attained 16.4~MW (55.9\%), exceeding the sum of its individual components.
\begin{figure}[htbp]
    \centering
    \includegraphics[width=0.65\linewidth]{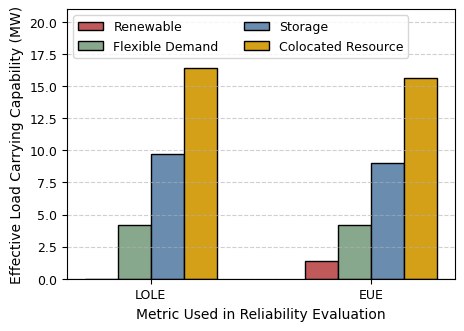}
    \vspace{-1em}
    \caption{ELCC of resources evaluated under LOLE- and EUE-based reliability criteria.}
    \label{fig:EUE_vs_LOLE}
\end{figure}

Fig.~\ref{fig:EUE_vs_LOLE} (right) presents ELCC values based on EUE. The renewable achieved an ELCC of 1.4~MW (10.1\%), while demand response and storage were credited 4.2~MW (100\%) and 9.0~MW (80.2\%), respectively. The colocated hydrogen system received 15.6~MW (53.2\%), again surpassing the sum of its individual components. These results demonstrate that the ELCC of a colocated system cannot be obtained by simply summing the ELCCs of its constituent resources.

\begin{figure}[htbp]
    \centering
    \includegraphics[width=0.65\linewidth]{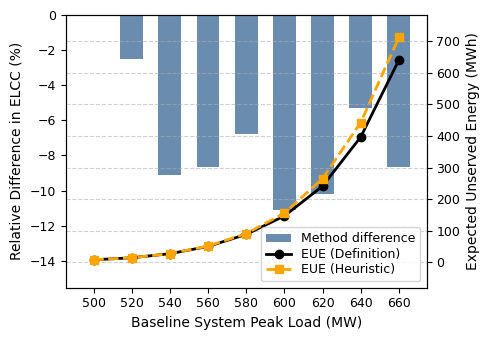}
    \vspace{-1em}
    \caption{ELCC differences of the colocated resource between the heuristic and proposed methods under different baseline system peak load levels.}
    \label{fig:Difference_ELCC}
\end{figure}

\noindent{\textbf{ELCC differences between heuristic and proposed methods:}}
Fig.~\ref{fig:Difference_ELCC} compares the ELCC of the colocated hydrogen system under a fixed scheduling priority defined by the operator~\cite{PJM_Manual20A} with the ELCC obtained from the proposed risk-minimization procedure. In the simulation, the baseline system load profile was uniformly scaled to achieve different peak load levels, while the installed generation capacity remained unchanged. The bars show the 
relative error, $\eta$, of the heuristic ELCC evaluation based on the EUE reliability criterion.  Specifically, given the heuristic ELCC estimate ($\hat{\mbox{\small ELCC}}$) and 
the ELCC by definition ($\mbox{{\small ELCC}}^*$) obtained via risk minimization,
\begin{align*}
\eta=(\hat{\mbox{\small ELCC}}-\mbox{{\small ELCC}}^* )/(\mbox{{\small ELCC}}^*).
\end{align*}
We observed that, except at a system peak load of 500~MW, the heuristic method consistently underestimated ELCC, with errors ranging from just over 2\% to more than 10\%. Such discrepancies could result in substantial monetary losses for the colocated resource.

Fig.~\ref{fig:Difference_ELCC} also shows the EUE computed using the heuristic and the proposed approaches. Across the full range of peak loads, the heuristic method consistently overestimated EUE. Both methods exhibited a convexly increasing trend, reflecting the growing difficulty of meeting demand during scarcity hours under more stressed system conditions.



\section{Conclusions}\label{sec:Conclusion}
Colocating large loads with on-site renewable and storage resources is a promising solution for managing grid demand surges. This paper highlights the underexplored benefits of colocation for resource adequacy in capacity markets.
We propose an efficient risk-minimization approach for accurately evaluating the ELCC of colocated resources, generalizing existing methods and applicable to large loads such as data centers and manufacturing facilities. Our case study demonstrates that simple approaches—such as summing individual ELCCs or using heuristic estimates—can produce significant errors when applied to colocated systems, underscoring the need for rigorous evaluation methods.

\section*{Acknowledgments}
The authors thank Prof. Le Xie and Dr. Tongxin Zheng for insightful discussions on capacity accreditation methodologies.

{
\bibliographystyle{IEEEtran}
\bibliography{BIB}

@ARTICLE{Evans19TPS,
  author={Evans, Michael P. and Tindemans, Simon H. and Angeli, David},
  journal={IEEE Transactions on Power Systems}, 
  title={Minimizing Unserved Energy Using Heterogeneous Storage Units}, 
  year={2019},
  volume={34},
  number={5},
  pages={3647-3656},
  doi={10.1109/TPWRS.2019.2910388}}

@article{Edwards17SEGN,
title = {Assessing the contribution of nightly rechargeable grid-scale storage to generation capacity adequacy},
journal = {Sustainable Energy, Grids and Networks},
volume = {12},
pages = {69-81},
year = {2017},
issn = {2352-4677},
doi = {https://doi.org/10.1016/j.segan.2017.09.005},
author = {Gruffudd Edwards and Sarah Sheehy and Chris J. Dent and Matthias C.M. Troffaes}
}

@article{Garver:1966,
	title = {Effective Load Carrying Capability of Generating Units},
	volume = {PAS-85},
	issn = {0018-9510},
	doi = {10.1109/TPAS.1966.291652},
	number = {8},
	journal = {IEEE Transactions on Power Apparatus and Systems},
	author = {Garver, L. L.},
	month = aug,
	year = {1966},
	pages = {910--919},
}

@ARTICLE{Sioshansi14TPS,
  author={Sioshansi, Ramteen and Madaeni, Seyed Hossein and Denholm, Paul},
  journal={IEEE Transactions on Power Systems}, 
  title={A Dynamic Programming Approach to Estimate the Capacity Value of Energy Storage}, 
  year={2014},
  volume={29},
  number={1},
  pages={395-403},
  doi={10.1109/TPWRS.2013.2279839}}

@article{Wen&Song:2023,
title = {{ELCC}-based capacity value estimation of combined wind - storage system using {IPSO} algorithm},
journal = {Energy},
volume = {263},
pages = {125784},
year = {2023},
issn = {0360-5442},
doi = {https://doi.org/10.1016/j.energy.2022.125784},
author = {Lei Wen and Qianqian Song}
}

@report{McKinsey:24,
  title        = {Global Energy Perspective 2024},
  author       = {{McKinsey \& Company}},
  year         = {2024},
  month        = sep,
  note         = {Report},
  url          = {https://www.mckinsey.com/~/media/mckinsey/industries/energy%20and%20materials/our%20insights/global%20energy%20perspective%202024/global-energy-perspective-2024.pdf}
}

@report{Schlag&etal:20E3Rpt,
  title        = {Capacity and Reliability Planning in the Era of Decarbonization: Practical Application of Effective Load Carrying Capability in Resource Adequacy},
  author       = {Nick Schlag and Zach Ming and Arne Olson and Lakshmi Alagappan and Ben Carron and Kevin Steinberger and Huai Jiang},
  institution  = {Energy and Environmental Economics},
  month        = aug,
  year         = {2020},
  note         = {Report},
  url          = {https://www.ethree.com/wp-content/uploads/2020/08/E3-Practical-Application-of-ELCC.pdf}
}

@report{Stenclik&Goggin&Ela&Ahlstrom:ESIG22Rpt,
  title        = {Unlocking the Flexibility of Hybrid Resources: A Report of the Hybrid Resources Task Force  of the Energy Systems Integration Group},
  author       = {D. Stenclik and M. Goggin and E. Ela and M. Ahlstrom},
  institution  = {Energy Systems Integration Group},
  month        = mar,
  year         = {2022},
  note         = {Report},
  url          = {https://www.esig.energy/wp-content/uploads/2022/03/ESIG-Hybrid-Resources-report-2022.pdf}
}

@techreport{NERC_2011,
  author       = {{North American Electric Reliability Corporation (NERC)}},
  title        = {Methods to Model and Calculate Capacity Contributions of Variable Generation for Resource Adequacy Planning},
  year         = {2011},
  month        = {March},
  number       = {IVGTF 1-2},
  url          = {https://www.nerc.com/pa/RAPA/ra/Reliability%20Assessments%20DL/IVGTF1-2.pdf}
}

@techreport{PJM_Manual20A,
  author       = {{PJM Interconnection LLC}},
  title        = {{PJM} Manual 20A: Resource Adequacy Analysis},
  number       = {Manual 20A},
  year         = {2025},
  month        = may,          
  url          = {https://www.pjm.com/-/media/DotCom/documents/manuals/m20a.ashx},
}

@misc{zhao2024,
  author       = {Feng Zhao},
  title        = {A {MRI}-Based Resource Capacity Accreditation},
  year         = {2024},
  institution = {{ISO New England}},
  month        = {July},
  url          = {https://www.ferc.gov/media/presentation-mri-based-resource-capacity-accreditation-feng-zhao-iso-new-england-holyoke-ma},
}

@techreport{Ericson&etal:22NREL,
 title        = {Influence of Hybridization on the Capacity Value of {PV} and Battery Resources},
 author       = {Sean Ericson and Sam Koebrich and Sarah Awara and Anna Schleifer and Jenny Heeter and Karlynn Cory and Caitlin Murphy and Paul Denholm},
 institution  = {National Renewable Energy Laboratory (NREL)},
 year         = {2022},
 number       = {NREL/TP-5R00-75864},
 url          = {https://www.nrel.gov/docs/fy22osti/75864.pdf},
 doi          = {10.2172/1854329}
}

@article{Glenk&Reichelstein:19Nature,
title={Economics of converting renewable power to hydrogen},
volume={4},
DOI={10.1038/s41560-019-0326-1},
number={3},
journal={Nature Energy},
author={Glenk, Gunther and Reichelstein, Stefan},
year={2019},
month={Feb},
pages={216-222}
}

@article{Wu&etal:Energy22,
title = {A techno-economic assessment framework for hydrogen energy storage toward multiple energy delivery pathways and grid services},
journal = {Energy},
volume = {249},
pages = {123638},
year = {2022},
issn = {0360-5442},
author = {Di Wu and Dexin Wang and Thiagarajan Ramachandran and Jamie Holladay},
}

@misc{Li&etal:25arxiv,
      title={Renewable-Colocated Green Hydrogen Production: Optimality, Profitability, and Policy Impacts}, 
      author={Siying Li and Lang Tong and Timothy Mount and Kanchan Upadhyay and Harris Eisenhardt and Pradip Kumar},
      year={2025},
      eprint={2504.18368},
      archivePrefix={arXiv},
      primaryClass={eess.SY},
      url={https://arxiv.org/abs/2504.18368}, 
}

@misc{Li&etal:26PESIM,
      title={Energy Management for Renewable-Colocated Artificial Intelligence Data Centers}, 
      author={Siying Li and Lang Tong and Timothy D. Mount},
      year={2025},
      eprint={2507.08011},
      archivePrefix={arXiv},
      primaryClass={math.OC},
      url={https://arxiv.org/abs/2507.08011}, 
}

@misc{NYISO:Data,
author = {{New York ISO}},
  title =        {{New York ISO} CUSTOM REPORTS},
  url = {https://www.nyiso.com/custom-reports?report=rt_lbmp_zonal},
  year =         {},
  month =        {}
}
}


\end{document}